\documentclass[11pt,letterpaper,leqno]{article}

\usepackage{amsmath}
\usepackage{epsfig,graphics}
\usepackage{etoolbox}
\usepackage{palatino}
\usepackage{subcaption}
\usepackage{multicol}
\usepackage[hyperfootnotes=false]{hyperref} 
\hypersetup{pdfborder=0 0 0} 
\hypersetup{colorlinks=true, urlcolor=MidnightBlue, citecolor=MidnightBlue, linkcolor=MidnightBlue}
\usepackage{mathrsfs}
\usepackage{multirow}
\usepackage{graphicx}
\usepackage{caption}
\usepackage{subcaption}
\usepackage[bottom,multiple]{footmisc}
\usepackage{ragged2e}
\def\citeapos#1{\citeauthor{#1}'s (\citeyear{#1})}

\usepackage[usenames,dvipsnames,svgnames,table]{xcolor}

\makeatletter
\def\@footnotecolor{gray!70!black}
\define@key{Hyp}{footnotecolor}{%
 \HyColor@HyperrefColor{#1}\@footnotecolor%
}
\patchcmd{\@footnotemark}{\hyper@linkstart{link}}{\hyper@linkstart{footnote}}{}{}
\makeatother
\hypersetup{footnotecolor=black}

\makeatletter
\renewcommand\@biblabel[1]{}

\makeatother

\usepackage{amsthm, thmtools}
\usepackage{
nameref,
hyperref,
cleveref,
}

\usepackage{amscd}
\usepackage{amsthm}
\usepackage{amsfonts}
\usepackage{tikz}
\usepackage{float}
\usetikzlibrary{calc,intersections,through,backgrounds,shadings}
\usetikzlibrary{shapes.geometric}
\usetikzlibrary{through}
\usetikzlibrary{arrows.meta}
\usepackage{multirow}
\usepackage{mathtools}
\usepackage{amssymb}
\usepackage{caption}
\usepackage{xfrac}  
\usepackage[margin=.9in]{geometry}

\newcommand\nc{\newcommand}
\nc\on{\operatorname}

\theoremstyle{plain}

\newtheorem{prop}{Proposition}
\newtheorem{lem}{Lemma}

\usepackage{pgfplots}

\theoremstyle{definition}

\usepackage{epigraph}
\usepackage{natbib}
\usepackage{graphicx}
\usepackage{placeins}
\graphicspath{ {images/} }
\usepackage{etoolbox}
\usepackage{amsmath}

\usepackage[utf8]{inputenc}
\usepackage[english]{babel}
 \usepackage{bm}
\usetikzlibrary{shapes.geometric}
{ }

\newtheorem*{mydef1}{Definition - Complements and Substitutes}
\newtheorem*{mydef2}{Definition - Supermodularity and Submodularity}
\newtheorem*{mydef3}{Definition - Gross Complements and Gross Substitutes}
\newtheorem*{myrem}{Remark 1}
\newtheorem*{myrem2}{Remark 2}

\usepackage{enumitem}
\usepackage{chngcntr}
\usepackage{apptools}
\AtAppendix{\counterwithin{prop}{section}}
\AtAppendix{\counterwithin{lem}{section}}
\AtAppendix{\counterwithin{cor}{section}}
\AtAppendix{\counterwithin{cl}{section}}
\newcommand{\bx}{\boldsymbol{x}}
\newcommand{\bv}{\boldsymbol{v}}
\newcommand{\bpi}{\boldsymbol{\pi}}
\newcommand{\bt}{\boldsymbol{t}}
\newcommand{\bp}{\boldsymbol{p}}
\newcommand{\bq}{\boldsymbol{q}}
\DeclareMathOperator*{\argmaxA}{arg\,max}

\begin{document}

\title{Revisiting the Impact of Upstream Mergers with Downstream Complements and Substitutes}
\author{Enrique Ide\thanks{IESE Business School, Carrer de Arn\'{u}s i de Gar\'{i} 3-7, 08034 Barcelona, Spain (email: eide@iese.edu). I am deeply grateful for the invaluable comments and suggestions provided by Nano Barahona, Daniel Barron, Juan-Pablo Montero, Sebasti\'{a}n Otero, Eduard Talam\`{a}s, Heski Bar-Isaac (the Editor), three anonymous referees, and seminar participants at Paris-Saclay. I am also thankful for the support provided by the MCIN/AEI /10.13039/501100011033/ FEDER, UE for the project with Ref. PID2021-123113NB-I00. I declare that I have no relevant or material financial interests that relate to the research described in this paper. A previous version of this paper circulated under the title ``Cross-Market Mergers with Common Customers: When (and Why) Do They Increase Negotiated Prices?''}}
\date{\today }
\maketitle

\begin{abstract} I examine how upstream mergers affect negotiated prices when suppliers bargain with a monopoly intermediary selling products to final consumers. Conventional wisdom holds that such transactions lower negotiated prices when the products are complements for consumers and raise them when they are substitutes. The idea is that consumer demand relationships carry over to upstream negotiations, where mergers between complements weaken the suppliers’ bargaining leverage, while mergers between substitutes strengthen it. I challenge this view, showing that it breaks down when the intermediary sells products beyond those of the merging suppliers. In such cases, the merging suppliers' products may act as substitutes for the intermediary even if they are complements for consumers, or as complements for the intermediary even if they are substitutes for consumers. These findings show that upstream conglomerate mergers can raise prices without foreclosure or monopolisation and help explain buyer-specific price effects resulting from such mergers. \end{abstract}

\vspace{5mm}

\noindent Keywords: Upstream Mergers, Conglomerate Mergers, Complements and Substitutes, Bargaining.

\noindent \textit{JEL} classifications: D43, L13, L40.

\newpage

\section{Introduction}

Consider a merger between two upstream suppliers negotiating with a monopoly intermediary selling products to final consumers. Conventional wisdom holds that such transactions reduce negotiated prices when the products are complements for consumers and raise them when they are substitutes. This reasoning reflects the idea that downstream complementarities or substitutabilities carry over to upstream negotiations, where a merger of complements weakens suppliers’ bargaining leverage and a merger of substitutes strengthens it \citep{inderst,shaffer2,nevo,gautam,baye}. 

This paper challenges that view. I show that the conventional wisdom holds only when the intermediary’s portfolio consists exclusively of the merging suppliers’ products. When the portfolio includes other products, the link between downstream demand relationships and upstream bargaining breaks down. In such cases, the merging products may act as substitutes for the intermediary even if they are complements for consumers, leading the merger to raise—not lower—negotiated prices. Conversely, they may act as complements for the intermediary even if they are substitutes for consumers, causing the merger to lower—not raise—negotiated prices.

To understand why, consider the following example. There are three products, and the intermediary is a monopoly retailer offering one-stop-shopping convenience. Consumers incur a positive shopping cost to visit the retailer, and this cost varies across the population. While the products are unrelated in consumption—neither direct substitutes nor complements—their demands are linked through one-stop-shopping economies. Specifically, all products are gross complements because a broader portfolio makes visiting the retailer more attractive. This, in turn, encourages more consumers to incur the shopping cost, increasing demand for all the products the retailer carries.

Because products 1 and 2 are gross complements, carrying one increases the retailer’s profit from the other by attracting more consumers. If these were the only items in the portfolio, they would also be complements from the retailer’s perspective, since the profit from adding both would exceed the sum of adding each individually. In this case, the conventional wisdom holds: downstream complementarities transfer upstream, and a merger between the suppliers of products 1 and 2 would lower negotiated prices.

However, when the retailer also carries product 3, an additional factor arises: because both products 1 and 2 help attract consumers to the store, each increases demand for product 3. As a result, products 1 and 2 can act as substitutes in generating demand for product 3—and thus in the profit spillovers they create. This occurs when the profit gain from product 3 generated by carrying both products is less than the sum of the gains from carrying each individually.

Consequently, due to this spillover effect, the retailer’s incremental profit from adding products 1 and 2 together may be smaller than the sum of the profits from adding each individually. In this case, products 1 and 2 are substitutes for the retailer in terms of profits, even though all three products are gross complements in the downstream market. When this happens, a merger between the suppliers of products 1 and 2 raises—rather than lowers—negotiated prices, contradicting conventional wisdom.

I generalize this insight to an arbitrary number of products, broader forms of downstream complementarities, and even to cases of downstream substitutabilities. In the latter case, I show that products 1 and 2 can be complements for the intermediary in terms of profit, even if all products are gross substitutes in the downstream market. For this to occur, products 1 and 2 must be complements in generating (negative) profit spillovers—meaning that the reduction in profit from the rest of the lineup caused by adding products 1 and 2 together is smaller than the sum of the reductions caused by adding each one individually.

Finally, I examine whether natural sufficient conditions exist for downstream demand relationships to carry over upstream. For downstream complementarities, the answer is affirmative: if all products are strict gross complements and their inverse demands are weakly supermodular, then the intermediary necessarily views them as strict complements in terms of profits. The reason is that the intermediary’s pre-optimisation profit function is then supermodular in quantities, and supermodularity is preserved under maximization. This result implies that a stronger condition than gross complementarity is needed for downstream complementarities to transfer upstream in the presence of third-party products. It also implies that under linear demands, any pair of products that are gross complements downstream are also complements for the retailer in terms of profits.

However, no analogous result holds for downstream substitutes: even if all products are strict gross substitutes and their inverse demand functions are weakly submodular, the products need not be substitutes for the retailer in terms of profits. This reflects an asymmetry between downstream complementarities and substitutabilities. The asymmetry arises because, although the intermediary’s pre-optimisation profit function is submodular under these conditions, submodularity is not preserved under maximization. As a result, it is possible to construct scenarios in which all products are strict gross substitutes and have weakly submodular inverse demands, yet some still act as complements for the retailer in terms of profits.

From an antitrust perspective, my results have important implications. First, they challenge the conventional wisdom asserting that downstream substitutabilities invariably exacerbate the competitive concerns of upstream mergers, while downstream complementarities mitigate them. For downstream substitutabilities, this misconception could lead to excessively stringent cost-synergy thresholds for merger approval, potentially blocking transactions that would enhance overall welfare. For downstream complementarities, it risks prematurely classifying a merger as efficiency-driven—without further evaluation—when the merger may instead be motivated by the prospect of securing higher negotiated prices.

Second, my findings call for a reassessment of how upstream conglomerate mergers—those involving suppliers of products that final consumers do not perceive as direct substitutes—are screened. To date, most theories stressing how portfolio effects can raise prices have relied on monopolisation or foreclosure arguments \citep[e.g.,][]{whinston,greenlee,molina,idemontero2024}. Hence, it is not surprising that antitrust authorities have been evaluating conglomerate mergers through those lenses. For instance, in its Guidelines for Non-Horizontal Mergers, the \cite{eu2008} states:
\begin{quote}
 The main concern in the context of conglomerate mergers is that of foreclosure. The combination of products...may confer on the merged entity the ability and incentive to leverage a
strong market position from one market to another by means of tying or bundling or other exclusionary practices.
\end{quote} However, my results show that upstream conglomerate mergers can raise prices without foreclosure or monopolisation, highlighting the need to revisit the evaluation criteria for these transactions.

Third, my results help explain buyer-specific price effects observed after upstream mergers. In the U.S. healthcare market, hospitals (“upstream suppliers”) negotiate prices with health insurers (“intermediaries”), who then offer plans to employers and households (“consumers”). After a merger between two competing hospitals in North Carolina, \cite{thompson} documents substantial variation in price changes among insurers: two insurers faced large price increases, one saw no change, and another experienced notable price decreases. Similarly, \citet{haas2009two} report insurer-specific price effects following two hospital mergers in Illinois.

Although neither study provides detailed insurer-level data---and therefore cannot directly confirm the exact reasons for insurer-specific differences---the observed variation is consistent with the portfolio mechanism proposed here:  differences in insurers’ networks likely affect whether the merging hospitals act as complements or substitutes for each insurer. Moreover, as I discuss in Section \ref{sec:implications}, alternative explanations struggle to account for this insurer-specific variation.

More broadly, the fact that upstream mergers can raise negotiated prices for some intermediaries while lowering them for others highlights their complex distributional consequences. Authorities should therefore recognize that merger effects may not be uniform across buyers. Accurately weighing these gains and losses typically requires a structural model, which can be costly to estimate. However, as I discuss in Section \ref{sec:implications}, consumer loss ratios—sometimes recoverable without a full structural model, using data from blackouts, delistings, or consumer surveys—can provide a useful first approximation for assessing whether a merger is likely to generate such mixed effects.

\subsection{Related Literature}

First and foremost, my paper contributes to the literature on the effects of upstream mergers when input prices are negotiated.\footnote{\cite{von}, \cite{dobson}, and \cite{chipty} analyse the impact of downstream mergers when input prices are negotiated, while  \cite{rogerson} and \cite{dopper} examine the effects of vertical integration.} This line of research originates with \cite{horn}, who examine upstream and downstream merger incentives in a setting where two downstream firms have interdependent demands. \cite{inderst} extend this work by studying how upstream and downstream mergers strategically influence the distribution of rents and incentives for technology adoption. Both articles find that upstream mergers increase the merging suppliers’ bargaining leverage when products are downstream substitutes but decrease it when the products are downstream complements.

More closely related is the work by \cite{baye}, who analyse a Hotelling-type model involving two suppliers, a monopoly retailer, and final consumers, some of whom are one-stop shoppers. They show that the presence of one-stop shoppers in the downstream market renders the suppliers’ products complements to the retailer. Consequently, an upstream merger of the two suppliers weakens the suppliers’ bargaining leverage, leading to lower negotiated prices.

My paper’s main departure from this line of work lies in considering intermediaries or retailers with portfolios containing more than two products.\footnote{There are several additional differences between these articles and my work. For example, \cite{horn} and \cite{baye} restrict their analyses to linear contracts, while \cite{inderst} focus exclusively on two-part tariffs/efficient contracts. In contrast, my results apply to both types of contracts. Furthermore, \cite{horn} features strategic interactions between the two downstream firms (their demands are interdependent), whereas I follow \cite{inderst} and \cite{baye} by abstracting from such interactions and instead focusing on a single monopoly retailer.} By moving beyond the two-product assumption prevalent in the prior literature, I demonstrate that the presence of third-party products can disrupt the transfer of downstream complementarities or substitutabilities to upstream negotiations. This overturns the conventional wisdom that upstream mergers reduce negotiated prices when the products are complements for consumers and increase prices when they are substitutes.

An important exception to the two-product focus in the literature is the work of \cite{shaffer2}. They analyse a model where \( n \geq 2 \) upstream suppliers sell their products through a single monopoly retailer, with all products being gross substitutes. On page 584, they claim that because products are gross substitutes, they are substitutes for the retailer, providing proof for that result on page 593. However, as I show in the Online Appendix, this conclusion does not hold in general; the proof in \cite{shaffer2} contains an error.

My paper also contributes to the literature on cross-market hospital mergers—mergers between hospitals operating in different geographic areas or serving different patient populations. These mergers have attracted attention because empirical studies in the U.S. find that they can raise the prices negotiated with health insurers \citep{lewis,dafny},\footnote{For instance, \cite{lewis} find that independent hospitals acquired by hospital systems operating in other markets increase insurers’ prices by 17\%. Similarly, \citet{dafny} document that cross-market hospital mergers are associated with 7–9\% higher prices for the merging entities.} even though hospitals offering non-overlapping services are likely complements for employers and households—the customers purchasing the insurance plans—and thus for insurers as well \citep{gautam}.\footnote{In the U.S. healthcare market, employers and households likely view hospitals that offer non-overlapping services as complements, as they seek to meet a range of healthcare needs for their employees or family members.}

Several theories have been proposed to explain this phenomenon. \cite{lewis} argue that price increases stem from greater post-merger bargaining power, captured by a higher bargaining weight in the Nash bargaining protocol. This increased bargaining strength may arise as larger hospital systems pool information about insurers or share the costs of building more sophisticated contract negotiation teams. \cite{peters} proposes a different mechanism: cross-market mergers can raise hospitals’ reservation payoffs when insurers compete by altering consumer behaviour after a negotiation impasse. In particular, more consumers switch to rival insurers post-merger than pre-merger, allowing the merged hospitals, which still contract with those rivals, to recapture a larger share of lost consumers. This reduces the cost of an impasse for hospitals and strengthens their bargaining position.

In contrast, I propose an alternative explanation: profit spillover effects on other hospitals within an insurer’s network can render hospitals substitutes for insurers, even if they are complements for employers and households. To demonstrate that this mechanism is different, I assume a single monopoly intermediary and hold bargaining weights constant across upstream market structures. This approach shuts down these alternative channels through which the merger could lead to price increases.

More closely connected to my work is the “Holes in the Network” (HiN) theory of harm, developed by \cite{vistnes} and \cite{dafny}. In essence, the HiN theory posits that two hospitals offering non-overlapping services—hospitals 1 and 2—can act as substitutes for an insurer in terms of profits if employers and households are willing to purchase an insurance plan that excludes either hospital, but not both (i.e., customers tolerate one but not two “holes” in the network). As I explain in Section \ref{sec:implications}, the HiN theory represents an extreme version of my argument, describing the scenario where hospitals 1 and 2 are “perfect substitutes” in generating profit spillovers to other hospitals in the insurer’s network. 

While insightful, the current presentation of the HiN theory has several limitations that this paper addresses. The central limitation is its inability to reconcile its findings with the broader literature: How can two hospitals be complements for employers and households—who value broad hospital networks—yet substitutes for insurers? By adopting a more general framework, my paper resolves this contradiction and shows when complementarities at the consumer level do, and do not, translate to the insurer (intermediary) level.

Moreover, the HiN theory focuses narrowly on downstream complementarities driven by one-stop-shopping economies—that is, the desire of employers and households to purchase plans covering a wide range of health needs. By contrast, my analysis shows that profit spillover effects can sever the link between downstream demand relationships and upstream bargaining outcomes across a wider range of settings, including cases involving downstream substitutabilities.

\section{A Simple Model} \label{sec:simple}

In this section, I present my main results using a simple model designed to clarify the underlying mechanism with maximum transparency. The model abstracts from retail and wholesale prices, as well as marginal costs, by assuming reduced-form payoffs. It focuses on a specific scenario: the merger of two upstream suppliers of consumer goods selling unrelated products through large retail chains that offer one-stop shopping to consumers. An example of such a merger is the 2005 transaction between Procter \& Gamble and Gillette (EU Case No. COMP/M.3732). In this setting, downstream complementarities arise solely from one-stop-shopping economies.

In Section \ref{sec:general}, I extend these results by incorporating firms’ pricing decisions and allowing for more general forms of complementarity and substitutability at the downstream level.

\subsection{The Setting} \label{sec:setting}

\textit{Setup.---} Consider a large monopoly retailer offering one-stop shopping convenience to a unit mass of infinitesimal consumers. The retailer can carry up to three products, \(i \in \{1, 2, 3\}\), which are neither direct substitutes nor technological complements (i.e., they are ``unrelated'' in consumption). Each product is exclusively produced by a single supplier, allowing products and suppliers to be referred to interchangeably. 

Suppliers can only reach final consumers indirectly via the retailer. Final consumers have a unit demand for each product. I assume---in a reduced-form fashion---that purchasing a unit of product $i$ provides the consumer with a surplus of \(v_i > 0\) and generates profits of \(\pi_i > 0\) for the retailer. For simplicity, suppliers do not directly profit from the sales of their products. Because products are unrelated, when the consumer buys multiple products, her surplus---and the retailer's profits---are additive over these products. 

Consumers must also incur a shopping cost $\xi \ge 0$ when visiting the retailer. This shopping cost varies among the population of consumers according to the cumulative distribution function (CDF) $G$. Final consumers observe the retailer’s product portfolio before deciding whether to incur this cost. The shopping cost can be understood literally as the opportunity cost of time spent navigating traffic, finding parking, and choosing products, or---taking a partial equilibrium view---as the value for the consumer of visiting its best alternative among the retailer's competitors.

I analyse the effects of a merger between suppliers 1 and 2 under the assumption that the retailer and suppliers negotiate lump-sum carriage fees using a Nash-in-Nash bargaining protocol. I denote by \( \beta \in (0, 1) \) the bargaining weight of the retailer in its negotiations with each supplier, where these weights are independent of the upstream market's ownership structure.

For independent suppliers, negotiations with the retailer proceed with the expectation that agreements will also be reached with other suppliers. For suppliers within a conglomerate, the process is similar but with a critical distinction: if the retailer fails to reach an agreement with any individual supplier within the conglomerate, it loses access to all products provided by the conglomerate. This is equivalent to the suppliers creating a new entity that centralizes the negotiations with the retailer.

\vspace{3mm}

\noindent \textit{Payoffs.---} Let $x_i \in \{0,1\}$ indicate whether the retailer carries product $i$ ($x_i=1$) or not ($x_i=0$). The retailer’s portfolio of products is thus represented by $\bx \equiv (x_1,x_2,x_3)$. Additionally, denote by $t_i$ the lump-sum fee paid by the retailer to supplier $i$ when an agreement is reached. For brevity, let $\bv \equiv (v_1,v_2,v_3)$, $\bpi \equiv (\pi_1,\pi_2,\pi_3)$, and $\bt \equiv (t_1,t_2,t_3)$. 

The indirect utility function of a consumer with a shopping cost of $\xi$, given the retailer’s portfolio $\bx$, is $U(\bx;\xi)=\max\{  \bx \cdot \bv - \xi, 0 \}$. The demand for product \(i\) is equal to the fraction of consumers who visit the retailer, conditional on the retailer carrying the product: \begin{equation*} \label{eq:aa1} D_i(\bx) = x_i \cdot \mathbb{P}\big(  \textstyle \bx \cdot \bv  \ge \xi \big) =x_i \cdot G( \bx \cdot \bv )\end{equation*}

Given that the retailer earns $ \bx \cdot \bpi $ per consumer that visits the store, its total profits before and after fixed fees are $\Pi(\bx)$ and $\Pi(\bx) - \bx \cdot \bt$, respectively, where: \begin{equation*} \label{eq:Pi} \Pi(\bx)= \big( \bx \cdot \bpi \big) \cdot G( \bx \cdot \bv )\end{equation*} Finally, the profits of supplier \(i\) are \(0\) before fixed fees and \(x_i t_i\) after fixed fees.

\subsection{A Brief Discussion} \label{sec:disc}

Before proceeding with the analysis, I introduce the key concepts of complementarity, supermodularity, and gross complementarity (along with their respective opposites) and show that the setting described exhibits downstream complementarities according to two notions: consumers’ utility is supermodular, and products are gross complements in demand. Additionally, I discuss several assumptions underlying the model and provide evidence that intermediaries who fail to reach an agreement with a conglomerate often lose access to all of the conglomerate’s products.

\vspace{3mm}

\noindent \textit{Complements and Substitutes: Some Key Definitions}.--- I begin by defining the concepts of complements and substitutes, supermodularity and submodularity, and gross complements and gross substitutes.

\begin{mydef1} Consider an arbitrary function $f(x_i,x_j,x_k)$, where $(x_i,x_j,x_k) \in \{0,1\}^3$. Products $i$ and $j$ are: \vspace{-2mm} \begin{itemize}[leftmargin=*,noitemsep]
\item Complements in terms of $f$ if the function $f(x_i,x_j,x_k)$ has increasing differences in $(x_i,x_j)$, i.e., if: \begin{equation*} \label{eq:inc_dif}  f(1,1,x_k)-f(0,1,x_k) \ge f(1,0,x_k)-f(0,0,x_k), \ \text{for any $x_k \in \{0,1\}$ } \end{equation*}
\item Substitutes in terms of $f$ if the function $f(x_i,x_j,x_k)$ has decreasing differences in $(x_i,x_j)$. \end{itemize} \end{mydef1}

Intuitively, products \(i\) and \(j\) are complements if the incremental gain in terms of \(f\) from carrying product \(i\) when product \(j\) is also present, \(f(1, 1, x_k) - f(0, 1, x_k)\), is greater than the gain when product \(j\) is not present, \(f(1, 0, x_k) - f(0, 0, x_k)\). The definition of substitutability is analogous: products \(i\) and \(j\) are substitutes if the incremental gain is smaller when product \(j\) is also present than when it is not.

I now introduce the stronger concepts of supermodularity and submodularity:\footnote{The following definitions are consistent with the standard lattice-theoretic definitions. Specifically, when a function’s domain is the product of completely ordered sets—such as \(\{0, 1\}\) or \(\mathbb{R}\)—the function is supermodular if and only if it exhibits increasing differences in any pair of variables \citep{vives2010}.}

\begin{mydef2} The function $f(x_i,x_j,x_k)$ is: \vspace{-2mm} \begin{itemize}[leftmargin=*,noitemsep]
\item Supermodular if it has increasing differences in all pairs of variables. 
\item Submodular if it has decreasing differences in all pairs of variables. \end{itemize} \end{mydef2}

As the definitions above show, the notions of complementarity and substitutability, as well as the concepts of supermodularity and submodularity, apply to any function $f$. Hence, they are applicable in a variety of contexts. For instance, products $1$ and $2$ are complements for a consumer in terms of utility when her indirect utility function $U(\bx;\xi)$ has increasing differences in $(x_1,x_2)$.\footnote{Increasing differences and supermodularity are cardinal properties. As such, they are not preserved under monotone transformations. However, both concepts can be weakened to purely ordinal conditions. The ordinal counterpart of increasing differences is single-crossing, and the ordinal counterpart of supermodularity is quasi-supermodularity \citep{vives1999}. In this paper, I work with the stronger concepts of increasing differences and supermodularity because they are more intuitive.} Similarly, the same products are substitutes for the retailer in terms of profits if $\Pi(\bx)$ has decreasing differences in $(x_1,x_2)$.

In contrast, the notions of \textit{gross complements} and \textit{gross substitutes} apply exclusively to products' demand functions. The standard definition is that product $i$ is a gross complement (substitute) of product $j \neq i$ if an increase in the price of $j$ decreases (increases) the demand for $i$. In the specific setting under consideration---with reduced-form payoffs---the following definition captures this idea:

\begin{mydef3} \color{white} $\quad$ \color{black} \vspace{-2mm} \begin{itemize}[leftmargin=*,noitemsep]
\item Product $i$ is a gross complement of product $j \neq i$ if the demand for $i$ increases when product \(j\) is also present, i.e., if $D_i(\boldsymbol{x})$ is increasing in $x_j \in \{0,1\}$.
\item Product $i$ is a gross substitute of product $j \neq i$ if the demand for $i$ decreases when product \(j\) is also present. \end{itemize} \end{mydef3}

\vspace{1mm}

\noindent \textit{Downstream Complementarities}.--- Complementarity in utility and gross complementarity in demand are two related but distinct notions of downstream complementarities. The first reflects that the marginal utility of a product increases when another product is present, while the second captures that the demand for one product rises when another product is available.\footnote{Complementarity in utility is sometimes called ``q-complementarity'' or ``Edgeworth-Pareto complementarity,'' while gross complementarity is sometimes referred to as ``p-complementarity'' \citep{hicks,seidman}.} 

In the setting described in Section \ref{sec:setting}, both notions of downstream complementarities are present: the indirect utility of a consumer with shopping cost \(\xi\), \(U(\bx; \xi)\), is supermodular in \(\bx\), and all products are gross complements (i.e., the demand for a product increases with the size of the retailer’s portfolio). Intuitively, these downstream complementarities arise from one-stop shopping economies.

\vspace{3mm}

\noindent \textit{Why Assume a Monopoly Retailer?}.--- The primary aim of this paper is to challenge the prevailing view that upstream mergers weaken suppliers’ bargaining leverage with downstream complements and strengthen it with downstream substitutes. Assuming a single monopoly retailer allows me to isolate and focus on the core mechanism at play, avoiding the additional complexities introduced by competition among retailers. This approach also rules out alternative channels through which the merger could increase prices, such as the mechanism proposed by \cite{peters} discussed in the Related Literature section. For the same reason, I assume that Nash bargaining weights remain unchanged following the merger.

As such, this paper should not be seen as providing a comprehensive framework for evaluating upstream mergers, as it omits certain potential mechanisms through which mergers can affect prices. However, this does not, in my view, diminish the significance of its findings. The effects I uncover remain applicable even when additional complexities are introduced. Consequently, the conventional wisdom risks imposing overly stringent cost-synergy thresholds for merger approval in the case of downstream substitutes, even when other factors are considered. Conversely, in the case of downstream complements, it may lead authorities to prematurely classify a merger as efficiency-driven—without sufficient evaluation—when the merger may instead be motivated by the prospect of securing higher negotiated prices.

\vspace{3mm}

\noindent \textit{On Nash-in-Nash Bargaining}.--- I model upstream negotiations using the Nash-in-Nash bargaining protocol \citep{horn,co}. This protocol has become increasingly prominent in the theoretical and empirical literature analyzing the effect of upstream negotiations on market outcomes \citep[e.g.,][]{dobson,ho}. It has also become an integral part of authorities' toolkit when evaluating upstream mergers of suppliers and vertical mergers \citep[e.g.,][]{farrell,shapiro,rogerson}.

One potential concern with this solution concept is that it may not be suitable when suppliers’ products are strong complements for the retailer \citep{co}. However, in the Online Appendix, I show that my main result also holds when upstream negotiations are modeled using the Shapley Value. This protocol does not have the same problems as the Nash-in-Nash solution when strong complementarities are present and can be microfounded in a noncooperative framework, as demonstrated by \cite{stole} and \cite{inderst}.

\vspace{3mm}

\noindent \textit{Disagreement Outcomes in Conglomerate-Intermediary Bargaining}.---The model assumes that if the retailer fails to reach an agreement with any individual supplier within a conglomerate, it loses access to all the conglomerate's products. This assumption aligns with how real-world negotiations between consumer goods manufacturers and retailers typically unfold.

For instance, in 2009, when Delhaize, a Belgian supermarket, failed to reach an agreement with Unilever, it was forced to stop stocking nearly 300 Unilever products \citep{reuters}. More recently, in January 2024, Carrefour delisted all PepsiCo products—such as Pepsi, Lipton Tea, and Lay’s crisps—in four European countries after negotiations between the two parties broke down over manufacturer-imposed price hikes \citep{reuters2}.

This ``all-or-nothing'' negotiation approach is also commonly observed in the U.S. healthcare market. When evaluating hospital mergers, the Federal Trade Commission (FTC) accounts for the fact that insurers often lose access to an entire hospital system if negotiations break down with any one hospital within it \citep[see ][p. 274]{farrell}. Consistent with this practice, major academic studies estimating the effects of policy interventions and counterfactuals in the industry also model bargaining on an all-or-nothing basis \citep[e.g.,][]{nevo2, lewis, ho}.

\vspace{3mm}

\subsection{Upstream Mergers with Downstream Complementarities Can Increase Negotiated Prices} \label{sec:merger}

I now analyse how the merger of suppliers 1 and 2 affects the terms negotiated in equilibrium. In the current setting, it is straightforward to show that the merger does not affect the negotiation between the retailer and supplier 3—the supplier not involved in the transaction. Hence, without loss of generality, I focus exclusively on how the merger affects the negotiations between the retailer and suppliers 1 and 2.

However, the inclusion or exclusion of product 3 in the retailer’s portfolio will play a crucial role in determining the merger’s impact, as will soon become apparent. I therefore compare two scenarios: one where product 3 is present (\(x_3 = 1\)) and one where it is absent (\(x_3 = 0\)), treating \(x_3 \in \{0, 1\}\) as a parameter.

Before the merger, the negotiated lump-sum fees paid by the retailer to suppliers $1$ and $2$ solve the following Nash bargaining programs: \begin{align*} & t_1^{\text{pre}} =  \argmaxA_{t_1} \left[ \Pi(1,1,x_3)- \Pi(0,1,x_3)- t_1 \right]^{\beta} \times t_1^{1-\beta} \\
& t_2^{\text{pre}} =  \argmaxA_{t_2} \left[ \Pi(1,1,x_3)- \Pi(1,0,x_3)- t_2 \right]^{\beta} \times t_2^{1-\beta} \end{align*} That is, $t_i^{\text{pre}}$ maximises the Nash product between the retailer's and supplier $i$'s gains from trading with each other, assuming the retailer also reaches an agreement with the other supplier. The combined premerger payments that suppliers $1$ and $2$ can jointly secure are, therefore, equal to: \begin{equation} \textstyle \label{eq:2} T^{\text{pre}}= \sum_{i =1,2} t_i^{\text{pre}} = (1-\beta)\left[ 2\Pi(1,1,x_3)-\Pi(0,1,x_3) -\Pi(1,0,x_3)\right] \end{equation}

Suppose now that suppliers $1$ and $2$ merge to form a conglomerate. The newly negotiated terms for each supplier within the conglomerate must solve the following reformulated Nash bargaining problem:\begin{equation*} \label{eq:post} \textstyle (t_1^{\text{post}},t_2^{\text{post}}) =  \argmaxA_{(t_{1},t_{2})} \left[  \Pi(1,1,x_3)- \Pi(0,0,x_3) - t_1 -t_2 \right]^{\beta} \times  \left[ t_1+t_2 \right]^{\beta}  \end{equation*}This new Nash program reflects that the merger changes the disagreement outcome: if the retailer and the conglomerate fail to reach an agreement, the retailer loses access to both products and avoids paying both fees. The resulting post-merger terms are: \begin{equation} \label{eq:2post}  \textstyle T^{\text{post}}= \sum_{i=1,2} t_i^{\text{post}} = (1-\beta)\left[ \Pi(1,1,x_3)-\Pi(0,0,x_3) \right]  \end{equation}

Comparing (\ref{eq:2}) and (\ref{eq:2post}), the total payment to the conglomerate will be greater than the sum of premerger payments (i.e., $T^{\text{pre}}<T^{\text{post}}$) if and only if the merging products are strict substitutes for the retailer in terms of profits, i.e., if $\Pi(1,1,x_3)-\Pi(0,1,x_3) < \Pi(1,0,x_3)-\Pi(0,0,x_3)$. Intuitively, when products are substitutes, an upstream merger increases the suppliers’ bargaining leverage. This is because failing to reach an agreement with both suppliers simultaneously results in a more significant reduction in the retailer’s profits than failing to reach an agreement with both suppliers individually: \begin{multline*} \Pi(1,1,x_3)-\Pi(0,1,x_3) < \Pi(1,0,x_3)-\Pi(0,0,x_3) \\ \Longleftrightarrow  \underbrace{[ \Pi(1,1,x_3)-\Pi(0,1,x_3) ]}_{\substack{\text{profits loss from not reaching} \\ \text{an agreement with supplier $1$}} } + \underbrace{[ \Pi(1,1,x_3)-\Pi(1,0,x_3) ]}_{ \substack{\text{profits loss from not reaching} \\ \text{an agreement with supplier $2$}} }  < \underbrace{\Pi(1,1,x_3)-\Pi(0,0,x_3)}_{\substack{\text{profits loss from not reaching an} \\ \text{agreement with both simultaneously}} }  \end{multline*}

By the same token, when products are complements in profits for the retailer, the merger weakens the suppliers’ bargaining leverage. This happens because failing to reach an agreement with both suppliers simultaneously results in a smaller reduction in the retailer’s profits compared to failing to reach an agreement with both suppliers when negotiating individually.

As noted by \citet{shaffer1}, \citet{inderst}, and \citet{shaffer2}, products that are gross substitutes in demand should also be substitutes in profits from the retailer’s perspective. The reason is that when two such products are added to the retailer’s portfolio, some consumers switch between them rather than generating entirely new sales. As a result, the incremental profit from offering both products is less than the sum of the profits from offering each one in isolation.

By the same logic, consumer-level complementarities should also carry over to the retailer: when products are gross complements in demand, offering both increases the likelihood of joint purchases, making the combined profit greater than the sum of the stand-alone profits \citep{inderst,baye}. 

The problem with the conventional wisdom, however, is that it overlooks how the inclusion of a product affects the spillovers created for other products in the retailer’s portfolio. As a result, downstream complementarities and substitutabilities do not always transfer upstream.

To see this formally, note that, in the current setting, products 1 and 2 are complements for the retailer in terms of profits, i.e., $\Pi(1,1,x_3)-\Pi(0,1,x_3) > \Pi(1,0,x_3)-\Pi(0,0,x_3)$, if and only if:\begin{multline} \label{eq:main} \pi_1\left[G(v_1+v_2+x_3 v_3) - G(v_1+x_3 v_3)\right]+ \pi_2\left[G(v_1+v_2+x_3 v_3) - G(v_2+x_3 v_3) \right] \\ > x_3 \pi_3 \left[G(v_1+x_3v_3)+G(v_2+x_3v_3) -G(v_1+v_2+x_3 v_3) - G(x_3 v_3) \right]  \end{multline}where I use the fact that $\Pi(x_1,x_2,x_3)= \big( \bx \cdot \bpi \big) \cdot G( \bx \cdot \bv )$. 

The left-hand side of (\ref{eq:main})—which is always positive—captures the effects described by the conventional wisdom: Because products 1 and 2 are complements to consumers, carrying product 1 increases the retailer’s profits from product 2 by increasing the number of consumers visiting the retailer (and vice versa). When the retailer’s portfolio consists exclusively of the merging suppliers’ products (i.e., \(x_3 = 0\)), this is the only force determining the complementarity or substitutability of products 1 and 2 for the retailer in terms of profits, as the right-hand side of (\ref{eq:main}) is zero. Thus, the conventional wisdom holds in this case: downstream complementarities transfer upstream, so an upstream merger always decreases negotiated prices. 

However, when the retailer carries product 3 (i.e., \( x_3 = 1 \)), an additional factor comes into play, as the right-hand side of equation (\ref{eq:main}) is no longer zero. This term—strictly positive when \( G(x) \) is strictly concave and strictly negative when \( G(x) \) is strictly convex—captures whether products 1 and 2 act as complements or substitutes in the profit ``spillovers'' they generate for product 3, the other product in the retailer’s lineup.

Formally, the profit spillovers of products 1 and 2 on product 3 are given by: \begin{equation} \label{eq:Sx1x2}
S_3(x_1, x_2) \equiv  \pi_3D_3(x_1,x_2,1)- \pi_3 D_3(0,0,1) = \pi_3[ G(x_1v_1+x_2v_2+v_3)-G(v_3)]
\end{equation}which captures how the presence or absence of products 1 and/or 2 alters the profit the retailer derives from product 3. In this setting—where downstream complementarities arise solely from one-stop-shopping economies—these spillovers are always positive, as expanding the retailer’s portfolio invariably increases consumer traffic.

However, despite these spillovers being positive, products 1 and 2 can still act as substitutes in generating them. This occurs when: \begin{multline} \label{eq:S3ccc} S_3(1, 1)-S_3(0,1) < S_3(1,0)-S_3(0,0) \\
\Longleftrightarrow 0< \underbrace{\pi_3 \left[G(v_1+v_3)+G(v_2+v_3) -G(v_1+v_2+v_3) - G(v_3) \right]}_{\text{the exact same expression as the right-hand side of (\ref{eq:main}) when $x_3=1$}}  \end{multline}That is, when the increase in the retailer’s profits from product 3, achieved by adding products 1 and 2 simultaneously, is smaller than the combined profit increase from adding each product individually. In this context, this occurs when \( G(x) \) is strictly concave, as it implies that adding product 1 on top of product 2 (or vice versa) does not increase consumer traffic significantly more than adding either product alone.

Thus, according to (\ref{eq:main}), when products 1 and 2 are substitutes in the spillovers they generate for product 3, and this substitutability is sufficiently strong, products 1 and 2 can become substitutes for the retailer in terms of profits, even though \textit{all three products are complements for consumers}. When this occurs, the merger of suppliers 1 and 2 results in higher---not lower---negotiated prices. 

I summarize the above discussion in the following proposition:

\begin{samepage}
\begin{prop} \label{prop:main} $\qquad$ \vspace{-2mm} \nopagebreak \begin{itemize}[leftmargin=*]
\item When the retailer carries only products 1 and 2, downstream complementarities necessarily transfer upstream. As a result, a merger between suppliers 1 and 2 always reduces negotiated prices.
\item When the retailer’s portfolio includes product 3, downstream complementarities do not necessarily transfer upstream. Consequently, a merger between suppliers 1 and 2 can lead to higher negotiated prices, even when all products are complements for consumers.
\end{itemize} \end{prop}
\end{samepage}

\begin{proof} Immediate from the discussion above. \end{proof}

\subsection{Implications} \label{sec:implications}

\noindent \textit{Where Does the Conventional Wisdom Come From?}--- As evident from condition (\ref{eq:main}), downstream complementarities necessarily transfer upstream when the retailer carries only the merging suppliers’ products. This fact can help explain the origins of the conventional wisdom described in the previous section. Indeed, as discussed in the Related Literature section, most theoretical models showing that downstream complementarities (or substitutabilities) transfer upstream typically either focus on two products \citep[e.g.,][]{shaffer1,inderst,baye}.\footnote{The result that downstream complementarities and substitutabilities transfer upstream also appears in \cite{horn}, who also consider only two products. However, their setting differs from mine as it includes strategic interaction in the downstream market, with two downstream firms facing interdependent downstream demands.}

\vspace{3mm}

\noindent \textit{The ``Holes in the Network'' Theory of Harm}.--- Proposition \ref{prop:main} is closely related to the “Holes in the Network” (HiN) theory of harm, developed to explain why mergers of hospitals operating in different geographic or diagnostic markets can sometimes raise negotiated prices \citep{vistnes,dafny}.

In essence, the HiN theory states that if employers or households (the “consumers”) value having multiple hospitals included in an insurer’s network, two hospitals offering non-overlapping services—hospitals 1 and 2—can act as substitutes for the insurer in terms of profits. This occurs when consumers are willing to purchase the insurance plan if either hospital 1 or 2 is excluded, but not if both are excluded (i.e., customers tolerate one but not two “holes” in the plan).

The HiN theory represents an extreme version of my argument. Indeed, the theory describes a situation where \(G(v_1 + v_2 + v_3) = G(v_1 + v_3) = G(v_2 + v_3) = 1\) and \(G(v_3) = 0\). Under these conditions, products 1 and 2 are always substitutes for the retailer in terms of profits:  \begin{multline*}  \underbrace{\pi_1\left[G(v_1+v_2+v_3) - G(v_1+v_3)\right]}_{=0}+ \underbrace{\pi_2\left[G(v_1+v_2+v_3) - G(v_2+v_3) \right]}_{=0}\\ < \underbrace{\pi_3 \left[G(v_1+v_3)+G(v_2+v_3) -G(v_1+v_2+v_3) - G(v_3) \right]}_{=\pi_3}  \end{multline*}

Intuitively, if consumers continue visiting the retailer when either product 1 or 2 is missing but stop visiting entirely when both products are absent, the two products function as perfect substitutes in generating spillovers—particularly store traffic—for product 3. Furthermore, because adding product $1$ or product $2$ to the portfolio does not affect store traffic when the other is already present, the presence of product 1 does not increase the profits of product 2, and vice versa. Consequently, in the HiN theory, the substitutability of products 1 and 2 in terms of profits for the retailer is driven solely by their substitutability in generating spillovers for product 3.

Although the HiN theory is insightful, its current presentation has several shortcomings, which this paper addresses. The main one is its inability to resolve the apparent contradiction with the conventional wisdom outlined earlier: How can two hospitals be complements for consumers—who value broad networks—yet substitutes for insurers? This paper provides a unifying framework that resolves this contradiction and clarifies the conditions under which complementarities at the consumer level extend to the intermediary level.

Moreover, the HiN theory focuses narrowly on downstream complementarities driven by one-stop-shopping economies—that is, the desire of employers and households to purchase plans covering a wide range of health needs. As I demonstrate in Section \ref{sec:general}, however, the insights extend to other types of complementarities as well. Notably, they even apply when products are gross substitutes, as I argue next.

 \vspace{3mm}

\noindent \textit{Upstream Mergers with Downstream Substitutabilities can Reduce Negotiated Prices}.--- Suppose product demands take the more general form \( D_i(\bx) \) for \( i = 1, 2, 3 \), rather than being driven by one-stop shoppers with unit demands, as in Sections \ref{sec:setting}--\ref{sec:merger}. Assume also that all products are gross substitutes, meaning that \( D_i(\bx) \) is decreasing in \( x_j \in \{0, 1\} \) for all \( i \neq j \). If the retailer generates a profit of \( \pi_i > 0 \) per unit of product \( i \) sold, its total profits before fixed fees are $\Pi(\bx)= \sum_{i} x_i \pi_i  D_i(\bx)$. 

In this context, the conventional wisdom asserts that products 1 and 2 must be substitutes for the retailer in terms of profits, given that they are gross substitutes in the downstream market. Consequently, a merger between suppliers 1 and 2 should always increase negotiated prices. 

However, the presence of product 3 in the retailer’s portfolio again disrupts this logic. Indeed, in this case, products 1 and 2 are substitutes in terms of profits, $\Pi(1,1,1)-\Pi(0,1,1) < \Pi(1,0,1)-\Pi(0,0,1)$, if and only if:\begin{multline} \label{eq:mainS} \underbrace{\pi_1\left[D_1(1,1,1) - D_1(1,0,1)\right]}_{<0}+ \underbrace{\pi_2\left[D_2(1,1,1) - D_2(0,1,1) \right]}_{<0} \\ < \underbrace{\pi_3 \left[D_3(1,0,1)+D_3(0,1,1)-D_3(1,1,1)-D_3(0,0,x_3) \right]}_{ \gtrless 0}  \end{multline}Hence, products 1 and 2 can act as complements for the intermediary in terms of profits, even though all three products are gross substitutes in the downstream market. In such a scenario, the merger leads to a reduction in negotiated prices.

For this to occur, products 1 and 2 must be complements in generating (negative) profit spillovers for product 3. Specifically, the reduction in the retailer’s profits from product 3, resulting from adding both products 1 and 2 to the lineup, must be smaller than the combined reduction in profits when adding each product individually. Mathematically, this condition is: \begin{multline*} S_3(1, 1)-S_3(0,1) > S_3(1,0)-S_3(0,0) \\ \Longleftrightarrow 0> \underbrace{\pi_3[D_3(1,1,1) -D_3(0,0,1) ]}_{ \substack{\text{reduction in 3's profits from} \\ \text{adding 1 and 2 simultaneously}} } - \underbrace{\pi_3[D_3(1,1,1) -D_3(1,0,1) ]}_{ \substack{\text{reduction in 3's profits from} \\ \text{adding only product 1 }} } - \underbrace{\pi_3[D_3(1,1,1) -D_3(0,1,1) ]}_{ \substack{\text{reduction in 3's profits from} \\ \text{adding only product 2 }} }  \end{multline*}

\vspace{3mm}

\noindent \textit{The Effects of an Upstream Merger can be Buyer-Specific}.--- My results help explain why upstream mergers can produce buyer-specific price effects. For instance, \cite{thompson} analyses the 1998 horizontal merger between New Hanover Regional Medical Center and Columbia Cape Fear Memorial Hospital in Wilmington, North Carolina. The study documents significant variation in post-merger price changes across insurers (also referred to as ``managed care organisations'' or MCOs): two insurers faced substantial price increases, one insurer’s prices remained stable, and another experienced notable price decreases.

Similarly, \citet{haas2009two} examine two horizontal hospital mergers in Chicago, Illinois, in 2000: the acquisition of Highland Park Hospital (HPH) by Evanston Northwestern Healthcare (ENH), and the merger between Provena St. Therese Medical Center (STMC) and Victory Memorial Hospital (VMH). After the ENH/HPH merger, four of the five active insurers in the market experienced statistically significant price increases, while the fifth avoided price increases altogether. Following the STMC/VMH merger, three of the same five insurers saw statistically significant price decreases, one insurer experienced no significant change, and one faced significant price increases.

Notably, the insurer that suffered the largest price increases after the ENH/HPH merger was also the only one to experience price increases following the STMC/VMH merger. This leads \citet[][p. 20]{haas2009two} to conclude: \begin{quote} One interesting aspect of the price change estimates from the MCO data is their variance across MCOs, suggesting that some MCOs may be more vulnerable to hospital mergers than others. \end{quote}

This observed variation aligns with the portfolio mechanism proposed here: differences in insurers’ networks likely determine whether merging hospitals act as complements or substitutes for each insurer. However, neither study provides detailed insurer-level data---to protect insurers' anonymity---and therefore cannot directly confirm this mechanism.\footnote{If insurer characteristics were available, we could test for portfolio effects in two ways. First, insurers with smaller networks should be more likely to experience price increases, since the merging hospitals represent a larger share of their network, limiting the potential for portfolio spillover effects. Second, as I discuss below, we could examine consumer loss ratios to test whether insurers experiencing lower post-merger prices were those for whom the merging hospitals functioned as complements in generating spillover effects—a necessary condition for the hospitals to be substitutes from the perspective of consumers but complements from the perspective of insurers.} Nevertheless, alternative explanations struggle to account for these patterns, especially in the New Hanover/Cape Fear and STMC/VMH mergers.

Indeed, both mergers were horizontal, involving hospitals that competed—at least to some extent—in the same geographic area and diagnostic market \citep{thompson,haas2009two}. Thus, the conventional view would predict price increases across all insurers, not insurer-specific price decreases. A similar prediction follows from \citeapos{lewis} argument that mergers enhance hospitals’ bargaining skills, and from \citeapos{peters} “recapture effect,” as mergers always increase recapture rates (see the Related Literature Section for an overview of these alternative mechanisms).

Hence, in the absence of portfolio effects, explaining why some insurers experienced price increases while others saw price decreases would require combining one or more of these theories with insurer-specific synergies. However, it is unclear why a hospital system would realize synergies with certain insurers but not others. A more straightforward explanation is that these buyer-specific effects are a natural consequence of the portfolio mechanism developed in this paper.

\vspace{3mm}

\noindent \textit{Implications for Merger Enforcement}.---My findings carry significant implications for the evaluation of upstream mergers when input prices are negotiated.

First, they challenge the conventional wisdom asserting that downstream substitutabilities invariably exacerbate the competitive concerns of upstream mergers, while downstream complementarities mitigate them. For downstream substitutabilities, this misconception could lead to excessively stringent cost-synergy thresholds for merger approval, potentially blocking transactions that would enhance overall welfare. For downstream complementarities, it risks prematurely classifying a merger as efficiency-driven—without further evaluation—when the merger may instead be motivated by the prospect of securing higher negotiated prices.

Second, my findings call for a reassessment of how upstream conglomerate mergers are evaluated. As noted in the Introduction, the prevailing approach primarily emphasizes risks such as anticompetitive foreclosure or monopolisation. For example, in its assessment of the 2005 merger between Procter \& Gamble and Gillette, the European Commission approved the merger after concluding that it would not allow the parties to impose weak brands on retailers, foreclose competitors from accessing retailers’ limited shelf space, or hinder market entry of new products through bundling practices.\footnote{The Commission also forced Procter \& Gamble to divest its battery toothbrush business ``SpinBrush,'' since there were significant horizontal overlaps at the consumer level between this brand and Gillette's ``OralB'' brand. } However, my results demonstrate that negotiated prices can increase through mechanisms unrelated to foreclosure, highlighting the need to revisit the evaluation criteria for these mergers.

Third, the fact that upstream mergers can raise negotiated prices for some intermediaries while lowering them for others implies that such mergers may have complex distributional consequences. Some intermediaries and their consumers may benefit, while others may be harmed. Authorities should therefore recognize that merger effects need not be uniform across buyers—and that accurately weighing these gains and losses typically requires a structural model.

Because estimating such models can be costly, it would be useful to first assess whether profit spillovers are likely to cause the merging products to act as complements for some intermediaries and substitutes for others. As I explain next, consumer loss ratios—sometimes estimable without a full structural model, using data from blackouts, delistings, or consumer surveys—can offer a practical first approximation for this purpose.

\vspace{3mm}

\noindent \textit{Consumer Loss Ratios to Measure Spillover Complementarities/Substitutabilities}.---My results suggests that analyzing consumer loss ratios---which measure the proportion of customers lost by the retailer after removing a product or set of products---can provide a rough approximation as to whether the merging products are complements or substitutes in spillovers for a given intermediary. Notably, consumer loss ratios have been applied in other antitrust contexts, including the evaluation of vertical mergers (e.g., \textit{U.S. v. AT\&T, et al.,} Civil Case No. 17-2511, RJL, D.D.C. Jun. 12, 2018).
 
Indeed, in the one-stop-shopping model of Sections \ref{sec:setting}-\ref{sec:merger}, note that $S_3(1, 1)-S_3(0,1) >S_3(1,0)-S_3(0,0)$ if and only if: \begin{multline*} \pi_3 \big[  \underbrace{\big(G(v_1+v_2+v_3)-G(v_2+v_3)\big)}_{ \substack{ \text{consumer loss ratio} \\ \text{of product 1  ($\mathrm{CL}_1$)} } } + \underbrace{\big( G(v_1+v_2+v_3)-G(v_1+v_3)\big)}_{ \substack{ \text{consumer loss ratio} \\ \text{of product 2  ($\mathrm{CL}_2$)} } } - \underbrace{\big(G(v_1+v_2+v_3)- G(v_3) \big)}_{ \substack{ \text{consumer loss ratio of products} \\ \text{1 and 2 combined ($\mathrm{CL}_{12}$)} } }\big] \\ > 0\end{multline*}Hence, if the sum of the consumer loss ratio of products 1 and 2 is greater than the consumer loss ratio of simultaneously losing both products---i.e., if $\mathrm{CL}_1+\mathrm{CL}_2>\mathrm{CL}_{12}$---the two products are complements in generating spillovers to product $3$.

Although a precise estimation of consumer loss ratios would likely require a structural model, rough approximations can be derived from episodes such as blackouts, delistings, or order freezes, which are becoming increasingly common.\footnote{In addition to the Delhaize and Carrefour examples discussed in Section \ref{sec:disc}, other notable delisting episodes or order freezes in the retail industry include the 2009 standoff between Costco and Coca-Cola \citep{retailwire}, and the Edeka Group's conflict with 17 international brands, including Procter \& Gamble, Mars, and Pepsi \citep{marketscreener}. Blackouts have also become increasingly common in the cable TV industry, as documented by \citet{frieden}.}  This was the approach apparently used in the \textit{U.S. v. AT\&T, et al.} case \citep[see][]{shapiro,carlton}.

An alternative is to use consumer surveys similar to those sometimes employed by antitrust authorities to calculate diversion ratios in horizontal merger analyses \citep{reynolds}. For instance, in the Somerfield/Morrisons case—a merger between two supermarkets—the UK Competition Commission (now the UK Competition and Markets Authority) asked shoppers ``what they would have done had the store been closed'' \citep[][p. 419]{reynolds}. These types of questions could be easily adapted to estimate loss ratios by asking shoppers, “What would you do if the store stopped carrying product $i$?''

Finally, while consumer loss ratios can provide useful evidence about whether products are complements or substitutes in generating spillovers, it is important to recognize their limitations. In more general environments—including those with additional products or endogenous retail prices—two qualifications arise. 

First, estimating consumer loss ratios remains feasible only when spillovers have a common source---for instance, when they depend on how the retailer’s portfolio influences overall consumer traffic.\footnote{Otherwise, determining whether products 1 and 2 are complements or substitutes in spillovers becomes significantly more complex. It requires estimating how the removal of each product affects the demand for all remaining products and weighting those effects by the retailer’s profit margins (see the Online Appendix for details).} Second, even when estimation is feasible, interpreting $\mathrm{CL}_1 + \mathrm{CL}_2 - \mathrm{CL}_{12}$ as a test for whether products are complements or substitutes in generating spillovers is only approximate. The issue is that intermediaries’ retail prices may change when one or both products are removed. As a result, the approximation is most reliable when the retailer does not adjust prices in response to the merger, or when any such adjustments are small. In the Online Appendix, I formalize these qualifications and identify conditions under which retail price responses are likely to be limited.

\section{The General Result} \label{sec:general}

I now generalize my results by incorporating firms’ pricing decisions and examining broader forms of downstream complementarities and substitutabilities. For simplicity, I assume efficient upstream negotiations, meaning that parties bargain over two-part tariffs. This assumption implies that upstream mergers affect only firms’ profits and not consumer surplus.

In the Online Appendix, I consider an extension where upstream negotiations are over linear contracts. In such a case, consumer welfare is also affected.

\subsection{The Model} \label{sec:modelgen}

\textit{Setup.---} There are \( n \geq 2 \) independent upstream suppliers, each offering a single differentiated product through a monopoly retailer. Supplier \( i \) faces a constant per-unit cost \( c_i \geq 0 \) for product $i$. The monopoly retailer purchases these goods from upstream suppliers and resells them to final consumers.

Downstream demands are represented by $\boldsymbol{D}(\bp) \equiv [D_1(\bp),...,D_n(\bp)]$, where \( D_i(\bp) \) denotes the demand for product \( i \), and \( \bp \equiv (p_1, \dots, p_n) \) is the vector of retail prices. Demands are continuously differentiable, downward-sloping with respect to their own price, and have the property that there exists a finite $\bar{p}_i>0$ such that $D_i(p_i=\bar{p}_i, \bp_{-i})=0$ for any $\bp_{-i}=(p_j)_{j \neq i}$. Additionally, the demand system is invertible, with inverse demands denoted by $\boldsymbol{P}(\bq)  \equiv [P_1(\bq),...,P_n(\bq)]$, where $\bq \equiv (q_1,...,q_n)$ is the vector of quantities sold in the downstream market. 

I assume that all products are either (i) gross complements or (ii) gross substitutes. This implies that \( \partial D_i / \partial p_j \leq 0 \) for all \( i \neq j \) in the case of complements, or  \( \partial D_i / \partial p_j \geq 0 \) for all \( i \neq j \) in the case of substitutes.\footnote{As noted by \citet[][p. 144]{vives1999}, if $\partial D_i/\partial p_j  \le 0$ for $ i \neq j$, then $\partial P_i/\partial q_j  \ge 0$ for $ i \neq j$ (although the converse is not true). Similarly, if $\partial D_i/\partial p_j  \ge 0$ for $ i \neq j$, then $\partial P_i/\partial q_j  \le 0$ for $ i \neq j$ (again, the converse is not true).} I focus on these two extreme cases to ensure that the results are not driven by a mix of substitution and complementarity patterns in the downstream market. Put differently, this section demonstrates that even if all products are gross complements, some of them may still act as substitutes for the retailer. Conversely, even if all products are gross substitutes, some of them may still function as complements for the retailer.

As in the simple model of Section \ref{sec:simple}, I analyse the effects of a merger between suppliers 1 and 2. I assume that the retailer and suppliers negotiate two-part tariff contracts \( (w_i, t_i) \in \mathbb{R}^2 \), where \( w_i \) is the per-unit wholesale price for product \( i \), and \( t_i \) is a lump-sum fee paid by the retailer to supplier \( i \). As before, I use Nash-in-Nash as the bargaining protocol, with \( \beta \in (0, 1) \) denoting the retailer’s bargaining weight in its negotiations with each supplier. Additionally, I continue to assume that if suppliers 1 and 2 form a conglomerate and either of them fails to reach an agreement with the retailer, the retailer loses access to all the conglomerate's products.

Under these assumptions, a well-known result in the literature is that negotiated wholesale prices equal marginal costs \citep{shaffer2}. This result reflects the fact that equilibrium wholesale prices are set to maximise total industry profits, with the surplus divided through fixed fees \( \bt = (t_1, \dots, t_n) \). Thus, without loss of generality, I set $w_i=c_i$ for all $i=1, \dots,n$ and focus on how the merger affects the negotiated fees. 

The timing of the game is as follows: At date 1, the retailer and the suppliers negotiate over fixed fees. At date 2, after observing its portfolio of products, the retailer sets the retail prices for the products it carries.

\vspace{3mm}

\noindent \textit{Payoffs.---} As in Section \ref{sec:simple}, let $x_i \in \{0,1\}$ indicate whether the retailer carries product $i$ ($x_i=1$) or not ($x_i=0$). The retailer’s profits before and after fixed fees are $\Pi^*(\bx)$ and $\Pi^*(\bx) - \bx \cdot \bt$, respectively, where: \begin{align*} \Pi^*(\bx) & =\textstyle \max_{\bq} \Big\{  \sum_i (P_i(\bq)-c_i)q_i  \ \, \text{s.t. $q_i=0$ if $x_i=0$ for $i=1,\dots,n$} \Big\} \end{align*}I assume that product demands are well-behaved, ensuring that the retailer’s problem has a unique and interior maximum for any portfolio $\bx=(x_1,\dots,x_n)$. The profits of supplier $i$ are \(0\) before fixed fees and \(x_i t_i\) after fixed fees. 

\vspace{3mm}

\noindent \textit{The Effects of the Merger on Negotiated Fees.---} Under the stated assumptions, it can be shown that the merger between suppliers 1 and 2 does not affect the terms negotiated between the retailer and suppliers \( i = 3, \dots, n \). Furthermore, following reasoning analogous to that in Section \ref{sec:merger}, the merger enables the merging parties to extract better terms, i.e., $t_1^{\text{pre}} +t_2^{\text{pre}} < t_1^{\text{post}} +t_2^{\text{post}}$, if and only if the products are strict substitutes for the retailer in terms of profits when the retailer carries all the other products, i.e., when: \[ \Pi^*(1,1, \boldsymbol{1}_{-12})-\Pi^*(0,1,\boldsymbol{1}_{-12}) < \Pi^*(1,0,\boldsymbol{1}_{-12})-\Pi^*(0,0,\boldsymbol{1}_{-12}) \]where $\boldsymbol{1}_{-12}$ is a $(n-2)\times 1$ vector containing only ones. Conversely, the merger forces the merging parties to accept worse terms if and only if the products are strict complements for the retailer in terms of profits when the retailer carries all the other products.

The main results of this section, presented next, are twofold: (i) demonstrating that demand complementarities or substitutabilities do not always transfer upstream to the retailer’s profit function, and (ii) providing sufficient conditions under which these demand relationships do transfer upstream. However, before presenting the main results, two brief remarks about this more general model are in order.

\begin{myrem} \normalfont This setting encompasses the downstream complementarities described in the model of Section \ref{sec:simple}. Indeed, suppose that (i) consumers have quasilinear preferences and sufficiently high income, (ii) products are unrelated in consumption, and (iii) consumers face a shopping cost \( \xi \) distributed according to the CDF $G$. Then \( D_i(\bp) = q_i(p_i) G\left(\sum_{i} v_i(p_i)\right) \), where \( q_i(p_i) \) represents the demand for product \( i \) conditional on the consumer visiting the retailer, and \( v_i(p_i) = \int_{p_i}^\infty q_i(x) dx \) is the consumer surplus obtained from purchasing \( q_i \) units of \( i \) at a price \( p_i \). Furthermore, assuming \( G \) has a probability density function \( g \), it follows that \( \partial D_i / \partial p_j = -q_j(p_j) q_i(p_i) g\left(\sum_{i} v_i(p_i)\right) < 0 \) for \( i \neq j \).

However, the current environment also accommodates more general preferences and broader forms of complementarities and substitutabilities at the downstream level. For instance, demands \( \boldsymbol{D}(\bp) \) could result from the maximization of a representative consumer with general preferences \( u(\bq) \) and income \( I \). In this case, demand complementarities may arise from complementarities in consumption or income effects. The same is true regarding demand substitutabilities. \end{myrem}

\begin{myrem2} \normalfont Unlike the simplified model in Section \ref{sec:simple}, this section imposes no assumptions about whether the products are complements or substitutes in consumers’ utility, nor about the utility function being supermodular or submodular. Instead, I work directly with the products’ demands, requiring only gross complementarity or gross substitutability. 

I follow this approach for two reasons. First, as explained in Section \ref{sec:merger}, the conventional wisdom depends solely on whether products are gross complements or gross substitutes in the downstream market.\footnote{Indeed, as explained in Section \ref{sec:merger}, ``products should be substitutes for the retailer in terms of profits when they are gross substitutes in demand. This is because the incremental profits from adding two gross substitutes to the retailer’s portfolio should be smaller than the sum of the profits from adding each product individually, as the sales of one product increase when the other is not stocked. By the same reasoning, consumer-level complementarities must also transfer upstream: when products are gross complements in demand, the retailer’s incremental profits from adding both products should be larger than the sum of the profits from adding them individually.'' } Second, it avoids the need to explicitly model consumers’ utility maximization, enabling the framework to accommodate a broader range of settings. \end{myrem2}

\subsection{Main Results} \label{sec:main}

I begin by generalizing Proposition \ref{prop:main} of Section \ref{sec:merger}:

\begin{prop} \label{prop:main2} \color{white} $\qquad$ \color{black} \vspace{-3mm} \begin{itemize}[leftmargin=*,noitemsep]
\item When $n=2$, products 1 and 2 are:  \vspace{-2mm} \begin{itemize}[leftmargin=*,noitemsep]
	\item Strict complements for the retailer in terms of profits if the products are strict gross complements.
	\item Strict substitutes for the retailer in terms of profits if the products are strict gross substitutes.	
	\end{itemize}
\item When $n>2$, products 1 and 2 can: \vspace{-2mm} \begin{itemize}[leftmargin=*,noitemsep]
	\item Be strict complements for the retailer in terms of profits, even when all products are strict gross substitutes.
	\item Be strict substitutes for the retailer in terms of profits, even when all products are strict gross complements.	
	\end{itemize}
\end{itemize} \end{prop}

\begin{proof} See Appendix A.1. \end{proof}

This result is closely related to the findings in Section \ref{sec:merger}. When the retailer carries only products 1 and 2, there are no profit spillovers to other products in the retailer’s lineup. In this case, the relationship between products 1 and 2 in the downstream market solely determines the upstream bargaining effects. In contrast, when the retailer carries additional products, the connection between downstream and upstream complementarities or substitutabilities breaks down, as profit spillovers on other products begin to influence outcomes.

To prove this last part, I construct an example. Specifically, in the proof of Proposition \ref{prop:main2}, I consider an scenario with $n=3$, $c_i=0$ for $i=1,2,3$, and demands given by: \begin{equation} \begin{split} \label{eq:demandsEx} & D_1(\bp)=(1-p_1)+b [ (1-p_2) + (1-p_3)]  \\
&  D_2(\bp)=(1-p_2)+b [ (1-p_1) + (1-p_3)]   \\
&  D_3(\bp)=(1-p_3)+\gamma \sqrt{(1-p_1) + (1-p_2)}  \end{split} \end{equation}

It is clear that if $b>0$ and $\gamma>0$, all three products are strict gross complements, while if $b<0$ and $\gamma<0$, all three products are strict gross substitutes. However, as I show in the proof of the proposition, products 1 and 2 are strict substitutes for the retailer in terms of profits when $\gamma=1/2$ and $b \to 0^+$, and are strict complements when $\gamma=-1/2$ and $b \to 0^{-}$.

Intuitively, the key to this example lies in the decreasing returns exhibited by the complementarities (\(\gamma > 0\)) or substitutabilities (\(\gamma < 0\)) of products 1 and 2 on the demand for product $3$, as reflected by the square root term in the expression for \(D_3(\bp)\). As a result, in the case of gross complements, adding either product 1 or 2 is sufficient to generate most of the positive spillovers to product 3, making products 1 and 2 substitutes for the retailer in terms of profits. Conversely, in the case of gross substitutes, adding products 1 and 2 simultaneously reduces product 3’s demand less than adding both products individually. This makes products 1 and 2 complements for the retailer in terms of profits.\footnote{In this more general setting, the profit spillover of products 1 and 2 on product 3 is formally defined as \(S_3(x_1, x_2) \equiv \pi_3^*(x_1, x_2, 1) - \pi_3^*(0, 0, 1)\), where \(\pi_3^*(x_1, x_2, 1)\) represents the retailer’s profits from product 3 in the optimum, given the portfolio \(x = (x_1, x_2, 1)\). Products 1 and 2 are substitutes in spillovers for product 3 when \(S_3(x_1, x_2)\) has decreasing differences in \((x_1, x_2)\), and they are complements in spillovers for product 3 when \(S_3(x_1, x_2)\) has increasing differences in \((x_1, x_2)\).}

As noted in the Related Literature section, the second part of Proposition \ref{prop:main2} contradicts a formal result from \cite{shaffer2}. In their analysis, they consider a setting with general demands where \(n \geq 2\) upstream suppliers sell their products through a single monopoly retailer, with all products being gross substitutes. On p. 584, they assert that because the products are gross substitutes, they will necessarily be substitutes for the retailer. However, as I show in the Online Appendix, there is a flaw in the proof of this result, which appears on p. 593 of their paper. 

I now present a sufficient condition to ensure that downstream complementarities transfer upstream. Unfortunately, as I explain below, a similar result does not exist for downstream substitutabilities.

\begin{prop} \label{prop:sup} If products are strict \textbf{gross complements} and \textbf{inverse demands are weakly supermodular} in $\bq$, then all products are strict complements for the retailer in terms of profits (i.e., $\Pi^*(\bx)$ is strictly supermodular). \end{prop}

\begin{proof} See Appendix A.2. \end{proof}

According to this proposition, when third-party products are present, downstream complementarities transfer upstream only if a stronger condition than gross complementarity is satisfied. Specifically, if inverse demands are interpreted as marginal valuations, the products must be gross complements, and the marginal valuations must be supermodular.\footnote{If consumers have quasilinear preferences, gross complementarity is a second-order property of the underlying utility function, while supermodularity of inverse demands is a third-order property. Thus, gross complementarity and supermodular inverse demands are related but distinct properties.} An important corollary of this proposition is that under linear demands, products are always complements for the retailer in terms of profits if they are gross complements in the downstream market. This follows from the fact that inverse demands are also linear and therefore weakly supermodular.

The proof of Proposition \ref{prop:sup} leverages the fact that if products are strict gross complements and inverse demands are weakly supermodular in $\bq$, then the retailer's profit function, $R(\bq) \equiv \sum_i (P_i(\bq)-c_i)q_i$, is strictly supermodular in $\bq$: \[ \frac{\partial^2 R}{ \partial q_i \partial q_{j}} = \underbrace{\frac{\partial P_i}{\partial q_j}}_{>0} + \underbrace{\frac{\partial P_{j}}{\partial q_{j}}}_{>0} + \sum_{m=1}^{n} q_m \underbrace{\frac{\partial P_m}{\partial q_i \partial q_j}}_{\ge 0} >0, \text{for any $i \neq j$} \]Since supermodularity is preserved under maximization \citep[][p. 26]{vives1999}, this implies that $M(q_1,q_2) \equiv \max_{\bq_{-12}} R(\bq)$ is supermodular in $(q_1,q_2)$. Defining $(q_1^*,q_2^*)$, $q_1^{\dagger}$, and $q_2^{\dagger}$ as the maximisers of $M(q_1,q_2)$, $M(q_1,0)$, and $M(0,q_2)$, respectively, the supermodularity of $M(q_1,q_2)$ is sufficient to ensure that (see the proof of Proposition \ref{prop:sup} for details): \[ \underbrace{M(q_1^*,q_2^*)}_{\Pi^*(1,1,\boldsymbol{1}_{-12})} - \underbrace{M(0,q_2^{\dagger})}_{\Pi^*(0,1,\boldsymbol{1}_{-12})}  >\underbrace{M(q_1^{\dagger},0)}_{\Pi^*(1,0,\boldsymbol{1}_{-12})}- \underbrace{M(0,0)}_{\Pi^*(0,0,\boldsymbol{1}_{-12})} \]That is, products 1 and 2 are necessarily complements for the retailer in terms of profits when the retailer carries all other products. The proof of Proposition \ref{prop:sup} extends this argument to any other combination of products the retailer may carry (i.e., for any $\bx_{-12} = (x_3,\dots,x_n)$). This demonstrates that, under the stated conditions \(\Pi^*(\bx)\) is strictly supermodular in \(\bx\).

Unfortunately, a similar result does not exist for downstream substitutabilities. Specifically, products are \textbf{not} necessarily substitutes for the retailer in profits, even if they are all strict gross substitutes and inverse demands are weakly submodular in quantities. This highlights an intriguing asymmetry between downstream complementarities and downstream substitutabilities in this context.

To understand the source of this asymmetry, observe that if products are strict gross substitutes and inverse demands are weakly submodular, then \(R(\bq)\) is also strictly submodular in \(\bq\):\[ \frac{\partial^2 R}{ \partial q_i \partial q_{j}} = \underbrace{\frac{\partial P_i}{\partial q_j}}_{<0} + \underbrace{\frac{\partial P_{j}}{\partial q_{j}}}_{<0} + \sum_{m=1}^{n} q_m \underbrace{\frac{\partial P_m}{\partial q_i \partial q_j}}_{\le 0} <0, \text{for any $i \neq j$} \]However, unlike supermodularity, submodularity is not preserved under maximization. As a result, it is possible to construct examples where all products are strict gross substitutes and inverse demands are weakly submodular in \(\bq\), yet some products remain complements for the retailer in terms of profits. I provide one such example in Appendix B.

\section{Conclusions}

In this paper, I challenge the conventional wisdom that upstream mergers of suppliers reduce negotiated prices when their products are complements for final consumers and increase prices when they are substitutes. I show that this logic holds only when the intermediary’s portfolio consists exclusively of the merging suppliers’ products. When the portfolio includes additional products, the merging suppliers’ products may act as substitutes for the intermediary, even if they are complements for consumers, leading the merger to increase negotiated prices. Similarly, the merging suppliers’ products may act as complements for the intermediary, even if they are substitutes for consumers, resulting in a reduction in negotiated prices.

From an antitrust perspective, these insights caution against the overgeneralization that downstream complementarities always mitigate competitive concerns while substitutabilities invariably exacerbate them. Such misconceptions may lead to misguided merger enforcement decisions, either blocking welfare-enhancing transactions or prematurely approving mergers with anti-competitive motivations. A reevaluation of merger assessment criteria is essential to account for the mechanisms uncovered in this paper.

Looking forward, future research could explore how the portfolio effects identified here influence overall market structure and technological adoption, as examined by \cite{inderst}, or their impact on vertical mergers, as studied by \cite{dopper}. These lines of research could deepen our understanding of the interplay between market structure, technology choice, and overall welfare when input prices are negotiated.

\newpage

\begin{appendix}

\begin{center} \Large \textsc{Appendix} \end{center}

\section{Proofs} \label{app:B}

\subsection{Proof of Proposition 2}

For ease of exposition, I split the proof into smaller lemmas:

\begin{lem} When $n=2$, products 1 and 2 are:  \vspace{-2mm} \begin{itemize}[leftmargin=*,noitemsep]
	\item Strict complements for the retailer in terms of profits if the products are strict gross complements.
	\item Strict substitutes for the retailer in terms of profits if the products are strict gross substitutes.	
	\end{itemize} \end{lem}
	
\begin{proof} Throughout the proof, I use the convention that \(D_1(p_1, p_2 = \infty)\) represents the demand for product 1 when product 2 is not stocked, and \(D_2(p_1 = \infty, p_2)\) represents the demand for product 2 when product 1 is not stocked.

Suppose first that products are strict gross complements. I want to show that $\Pi^*(1,1)>\Pi^*(1,0) + \Pi^*(0,1)$ (note that $\Pi^*(0,0)=0$ in this case). Let $(p_1^*,p_2^*) \equiv \argmaxA_{(p_1,p_2)}  \sum_{i=1,2} (p_i-c_i)D_i(p_1,p_2)$, and define $p_1^{\dagger}\equiv \argmaxA_{p_1} (p_1-c_1)D_1(p_1,p_2=\infty) $ and $p_2^{\dagger} \equiv \argmaxA_{p_2} (p_2-c_2)D_2(p_1=\infty,p_2)$. By definition of $(p_1^*,p_2^*)$, I have that $\sum_{i=1,2} (p_i^*-c_i)D_i(p_1^*,p_2^*) > \sum_{i=1,2} (p_i^{\dagger}-c_i)D_i(p_1^{\dagger},p_2^{\dagger})$. Moreover, because products are strict gross complements, then $D_1(p_1^{\dagger},p_2^{\dagger})>D_1(p_1^{\dagger},\infty)$ and $D_2(p_1^{\dagger},p_2^{\dagger})>D_2(\infty,p_2^{\dagger})$. Thus, \begin{multline*} \textstyle \Pi^*(1,1) = \sum_{i=1,2} (p_i^*-c_i)D_i(p_1^*,p_2^*) > \sum_{i=1,2} (p_i^{\dagger}-c_i)D_i(p_1^{\dagger},p_2^{\dagger})  \\ \textstyle > (p_1^{\dagger}-c_1)D_1(p_1^{\dagger},\infty)+ (p_2^{\dagger}-c_2)D_2(\infty,p_2^{\dagger}) = \Pi^*(1,0) + \Pi^*(0,1) \end{multline*}implying that the products are strict complements for the retailer in terms of profits.

Suppose instead that products are strict gross substitutes. I aim to show that $\Pi^*(1,1)<\Pi^*(1,0) + \Pi^*(0,1)$. Define $(p_1^*,p_2^*)$, and $p_i^{\dagger}$ for $i=1,2$ as above. I have that $\Pi^*(1,0) = (p_1^{\dagger}-c_1)D_1(p_1^{\dagger},\infty)>(p_1^{*}-c_1)D_1(p_1^{*},\infty) >(p_1^{*}-c_1)D_1(p_1^{*},p_2^*)$, where the first inequality follows from the definition of $p_1^{\dagger}$, and the second inequality because products are strict gross substitutes. Thus, $\Pi^*(1,0) > (p_1^{*}-c_1)D_1(p_1^{*},p_2^*)$. By an analogous argument for product 2, it follows that $ \Pi^*(0,1)> (p_2^{*}-c_2)D_2(p_1^{*},p_2^*)$. Summing these inequalities gives  $\Pi^*(1,0)+ \Pi^*(0,1)> \sum_{i=1,2} (p_i^*-c_i)D_i(p_1^*,p_2^*)=\Pi^*(1,1)$, establishing that the products are strict substitutes for the retailer in terms of profits. \end{proof}	

\begin{lem} When $n>2$, products 1 and 2 can: \vspace{-2mm} \begin{itemize}[leftmargin=*,noitemsep]
	\item Be strict substitutes for the retailer in terms of profits when the products are strict gross complements.
	\item Be strict complements for the retailer in terms of profits when the products are strict gross substitutes.	
	\end{itemize} \end{lem}
	
\begin{proof} I prove this lemma by constructing examples. Suppose that $n=3$, $c_i=0$ for $i=1,2,3$, and that demands are given as in equation (7) of the main text. When $i=1,2$, then $\partial D_i/\partial p_j = -b$ for $j \neq i$, whereas $\partial D_3/\partial p_j = -0.5\gamma /\sqrt{2-p_1-p_2}$ for $j=1,2$. Hence, if $b>0$ and $\gamma>0$, all three products are strict gross complements, while if $b<0$ and $\gamma<0$, all three products are strict gross substitutes.

Let $\Delta \equiv \Pi^*(1,1,1) +\Pi^*(0,0,1)-\Pi^*(1,0,1)-\Pi^*(0,1,1)$. I will show that $\Delta<0$ when $b=0$ and $\gamma=1/2$, which by continuity, implies that $\Delta<0$ when $b>0$ but small and $\gamma=1/2$. Hence, in that case, the three products are strict gross complements, but products 1 and 2 are strict substitutes for the retailer in terms of profits. Similarly, I will show that $\Delta>0$ when $b=0$ and $\gamma=-1/2$, which by continuity, implies that $\Delta>0$ when $b<0$ but small and $\gamma=-1/2$. Hence, the three products are strict gross substitutes, but products 1 and 2 are strict complements for the retailer in terms of profits.

When $b=0$, inverse demands are $P_1(\bq) =1-q_1$, $P_2(\bq)=1-q_2$, and $P_3(\bq)=1-q_3+\gamma \sqrt{q_1+q_2}$. Clearly, $\Pi^*(0,0,1)= \max_{q_3} q_3(1-q_3)=1/4$. Moreover, \[ \Pi^*(1,1,1)= \max_{\bq}\big\{ q_1(1-q_1)+q_2(1-q_2)+q_3(1-q_3+\gamma \sqrt{q_1+q_2})   \big\} \]The optimum for an arbitrary $\gamma$ is $q_1^*=q_2^*$ and $q_3^*= (1+\gamma \sqrt{2q_2^*})/2$, where $q_2^*$ satisfies:\footnote{When $\gamma<0$, this condition might have more than one solution in $\mathbb{R}_{\ge0}$. However, it is straightforward to identify the one that is indeed optimal.} \[ 1-2q_2^* + \frac{\gamma^2}{4} + \frac{\gamma \sqrt{2}}{8 \sqrt{q_2^*}} =0 \]When $\gamma=1/2$, the optimum is $q_1^*=q_2^*=0.589$ and $q_3^*=0.771$, so $\Pi^*(1,1,1; \gamma=1/2)=1.08$. When $\gamma=-1/2$, the optimum is $q_1^*=q_2^*=0.467$ and $q_3^*=0.259$, so $\Pi^*(1,1,1; \gamma=-1/2)=0.565$. 

Finally,  \[ \Pi^*(1,0,1)= \Pi^*(0,1,1)= \max_{(q_i,q_3)}\big\{ q_i(1-q_i)+q_3(1-q_3+\gamma \sqrt{q_i})   \big\} \]The optimum for an arbitrary $\gamma$ is $q_3^*=(1+\gamma \sqrt{q_i^*})/2$, where $q_i^*$ satisfies: \[ 1-2q_i^* + \frac{\gamma^2}{4} + \frac{\gamma}{4 \sqrt{q_i^*}} =0 \]When $\gamma=1/2$, the optimum is $q_i^*=0.611$ and $q_3^*=0.695$, so $\Pi^*(1,0,1; \gamma=1/2)= \Pi^*(0,1,1; \gamma=1/2)=0.721$. When $\gamma=-1/2$, the optimum is $q_i^*=0.437$ and $q_3^*=0.335$, so $\Pi^*(1,0,1; \gamma=-1/2)= \Pi^*(0,1,1; \gamma=-1/2)=0.358$. 

Thus, when $\gamma=1/2$, $\Delta = 1.08+0.25-2 \cdot (0.721)=-0.112<0$, whereas when $\gamma=-1/2$, $\Delta = 0.565+0.25-2 \cdot (0.358)=0.099$. \end{proof}

\subsection{Proof of Proposition 3}

I will show that $\Pi^*(x_1,x_2,\bx_{-12})$ exhibits strict increasing differences in $(x_1,x_2)$ for any $\bx_{-12}$. Since $x_1$, $x_2$, and $\bx_{-12}$ are arbitrary, this immediately establishes that $\Pi^*(\bx)$ has strict increasing differences in any pair $(x_i,x_j)$ with $i \neq j$ (put differently, I can always relabel products so that $i$ is ``product 1'' and $j$ is ``product 2'').

Fix $\bx_{-12}$ and assume, without loss of generality, that among the products in $\bx_{-12}$, the retailer does not carry the last $k \le n-2$ products (note that $\bx_{-12}$ has length $n-2$). If $k=n-2$, then the retailer only carries products 1 and 2. In this case,  $\Pi^*(x_1,x_2,\bx_{-12})$ exhibits strict increasing differences in $(x_1,x_2)$, following the exact same reasoning as in the proof of Proposition 2.

Now consider the case where $k<n-2$. Let $R(\bq) \equiv \sum_{i=1}^{n-k} (P_i(\bq)-c_i)q_i$, and define: \[ \textstyle M(q_1,q_2) \equiv \max_{(q_3,...,q_{n-k})}   R(q_1,q_2,q_3,\dots,q_{n-k},0,\dots,0)  \]Let $(q_1^*,q_2^*)$, $q_1^{\dagger}$, and $q_2^{\dagger}$ be the maximisers of $M(q_1,q_2)$, $M(q_1,0)$, and $M(0,q_2)$, respectively. By assumption, $q_1^*$, $q_2^*$, $q_1^{\dagger}$, and $q_2^{\dagger}$ are all strictly positive. Moreover, by construction, $\Pi^*(1,1,\bx_{-12})=M(q_1^*,q_2^*)$, $\Pi^*(1,0,\bx_{-12})=M(q_1^{\dagger},0)$, $\Pi^*(0,1,\bx_{-12})=M(0,q_2^{\dagger})$, and $\Pi^*(0,0,\bx_{-12})=M(0,0)$. Thus, $\Pi^*(x_1,x_2,\bx_{-12})$ has strict increasing differences in $(x_1,x_2)$ if and only if: \[ M(q_1^*,q_2^*)-M(0,q_2^{\dagger}) >M(q_1^{\dagger},0)- M(0,0)\]

To show this last point, note that because products are strict gross complements, i.e., $\partial D_i/\partial p_j  < 0$ for $ i \neq j$, it follows that $\partial P_i/\partial q_j  > 0$ for $ i \neq j$ \citet[][p. 144]{vives1999}. This observation, combined with the fact that inverse demands are weakly supermodular in $\bq$, imply that: \[  \frac{\partial^2 R}{ \partial q_i \partial q_{j}} = \underbrace{\frac{\partial P_i}{\partial q_i}}_{>0} + \underbrace{\frac{\partial P_{j}}{\partial q_{j}}}_{>0} + \sum_{m=1}^{n-k} q_m \underbrace{\frac{\partial P_m}{\partial q_i \partial q_j}}_{\ge 0} >0, \text{for any $i \neq j$ such that $i\le n-k$ and $j \le n-k$} \]That is, $R(q_1,q_2,q_3,\dots,q_{n-k},0,\dots,0)$ is strictly supermodular in $(q_1,q_2,q_3,\dots,q_{n-k})$. Consequently, since supermodularity is preserved under maximization, \(M(q_1, q_2)\) is strictly supermodular in \((q_1, q_2)\), meaning it has strictly increasing differences in \((q_1, q_2)\).

Now, by the definition of $(q_1^*,q_2^*)$, $M(q_1^*,q_2^*)>M(q_1^{\dagger},q_2^{\dagger})$. Moreover, since $q_1^{\dagger}>0$, $q_2^{\dagger}>0$, and \(M(q_1, q_2)\) is strictly supermodular in \((q_1, q_2)\), then $M(q_1^{\dagger},q_2^{\dagger})>M(q_1^{\dagger},0)+M(0,q_2^{\dagger})-M(0,0)$. Combining both inequalities, $M(q_1^*,q_2^*)>M(q_1^{\dagger},0)+M(0,q_2^{\dagger})-M(0,0)$, or, equivalently, $M(q_1^*,q_2^*)-M(0,q_2^{\dagger}) >M(q_1^{\dagger},0)- M(0,0)$. \hfill $\square$

\section{An Example Illustrating Why There is No Analog to Proposition 3 for Gross Substitutes} \label{app:C}

In the main text, I claimed that it is not possible to derive a result analogous to Proposition 3 for the case of strict gross substitutes. This is because submodularity is not preserved under maximization. Here, I provide a concrete example to illustrate this phenomenon. Specifically, I demonstrate that even if products are strict gross substitutes and inverse demands are weakly submodular in $\bq$, two products in the portfolio can still be complements for the retailer in terms of profits.

Suppose that $n=3$, $c_i=0$ for $i=1,2,3$, and that inverse demands are given by: \begin{equation*} \begin{split} & P_1(\bq)=1-q_1+b \ln(1+q_2)+\alpha q_3  \\
&  P_2(\bq)=1-q_2+b \ln(1+q_1)+\alpha q_3   \\
&  P_3(\bq)=1-q_3+\gamma (q_1+q_2)  \end{split} \end{equation*} Note that $\partial^2 P_i /\partial q_i \partial q_j =0$ for any $j \neq i$, and that $\partial^2 P_i /\partial q_k \partial q_j =0$ for any $ j \neq k \neq i$. Hence, these inverse demands are weakly submodular. Moreover, by the implicit function theorem: \begin{align*} & \textstyle  \frac{\partial D_i}{\partial p_j} = \frac{-(1+q_i)(b+ \alpha \gamma (1+q_j)) }{\varphi}, \ \text{for $i=1,2$ if $j=1,2$ and $j \neq i$} \\
& \textstyle  \frac{\partial D_i}{\partial p_3} =  \frac{-\alpha(1+q_i)(1+q_j+b)}{\varphi}, \ \text{for $i=1,2$}  \\
&  \textstyle  \frac{\partial D_3}{\partial p_i}=  \frac{-\gamma(1+q_i)(1+q_j+b)}{\varphi}, \ \text{for $i=1,2$} \end{align*}where $\varphi \equiv 1-b^2 + q_1+q_2+q_1q_2 -\alpha \gamma ( 2(1+q_1)(1+q_2) + b(2+q_1+q_2) )$. From these expressions, it is easy to see that if $|b| <1$ and $| \alpha | \to 0$, then all three products are strict gross substitutes when $\alpha$, $b$, and $\gamma$ are strictly negative. In contrast, all three products are strict gross complements when $\alpha$, $b$, and $\gamma$ are strictly positive.

Define then $M(q_1,q_2) \equiv \max_{q_3} \sum_{i=1}^{i=3} P_i(\bq)q_i$. Note that: \[ \textstyle \frac{\partial^2}{ \partial q_1 \partial q_2}\left(\sum_{i=1}^{i=3} P_i(\bq)q_i \right)= \frac{b}{1+q_1}+ \frac{b}{1+q_2} \quad \text{and} \quad  \frac{\partial^2}{ \partial q_i \partial q_3}\left(\sum_{i=1}^{i=3} P_i(\bq)q_i \right)=\alpha+\gamma, \ \text{for $i=1,2$} \]Thus, before maximizing its quantities, the retailer's profits are strictly submodular in $\bq$ if $\alpha$, $b$, and $\gamma$ are strictly negative, i.e., when all three products are strict gross substitutes. However, I also have that:\begin{equation*} \label{eq:M} \frac{\partial^2 M}{\partial q_1 \partial q_2} = \frac{\alpha^2}{2} + \alpha \gamma + \frac{b}{1+q_1}+ \frac{b}{1+q_2} + \frac{\gamma^2}{2} \end{equation*}which immediately implies that there exists $\alpha<0$, $b<0$, and $\gamma<0$ such that $\partial^2 M/\partial q_1 \partial q_2 >0$. For example, letting $\alpha<0$ but $\alpha \to 0$, $b=-1/8$, and $\gamma=-4/5$, I have:\[  \frac{\partial^2 M}{\partial q_1 \partial q_2} = \frac{8}{15} - \frac{1}{8(1+q_1)}- \frac{1}{8(1+q_2)} > \frac{7}{100} >0  \]Since $M(q_1,q_2)$ is strictly supermodular in $(q_1,q_2)$ this implies that when the retailer carries product 3, products 1 and 2 are strictly complements for the retailer in terms of profits (see the proof of Proposition 3). Indeed, in this particular case, $\Pi^*(1,1,1)+\Pi^*(0,0,1)-\Pi^*(1,0,1)-\Pi^*(0,1,1) = 0.015>0$.

Finally, note that this issue does not arise when $\alpha>0$, $b>0$, and $\gamma>0$, i.e., when products are strict gross complements. In this case, $ \sum_{i=1}^{i=3} P_i(\bq)q_i$ is always strictly supermodular in $\bq$ and so is $\partial^2 M/\partial q_1 \partial q_2 >0$. Indeed, as shown in Proposition 3, it is because $ \sum_{i=1}^{i=3} P_i(\bq)q_i$ is strictly supermodular in $\bq$ that $M(q_1,q_2)$ is strictly supermodular in $(q_1,q_2)$. 
\end{appendix}

{\bibliographystyle{ecta}
\bibliography{references}}

\end{document}